\documentclass[aps,pra,twocolumn,superscriptaddress,preprintnumbers,times]{revtex4-1}
\usepackage{epsfig,color,braket}
\usepackage{graphicx,amssymb,epsfig,amsmath,amsfonts,amsthm,bm,mathtools,mathrsfs,bbm}
\usepackage[colorlinks,citecolor=blue,linkcolor=magenta,urlcolor=blue]{hyperref}

\newcommand{\ketbra}[2]{\left|#1\middle\rangle\middle\langle#2\right|}

\newcommand{\be}{\begin{eqnarray}}
\newcommand{\ee}{\end{eqnarray}}

\DeclareMathOperator{\Tr}{Tr}
\newtheorem{theorem}{Theorem}
\newtheorem{lemma}{Lemma}

\begin{document}

\title{Sufficient criterion for guaranteeing that a two-qubit state is unsteerable}

\author{Joseph Bowles}
\affiliation{D\'epartement de Physique Th\'eorique, Universit\'e de Gen\`eve, 1211 Gen\`eve, Switzerland}
\author{Flavien Hirsch}
\affiliation{D\'epartement de Physique Th\'eorique, Universit\'e de Gen\`eve, 1211 Gen\`eve, Switzerland}
\author{Marco T\'ulio Quintino}
\affiliation{D\'epartement de Physique Th\'eorique, Universit\'e de Gen\`eve, 1211 Gen\`eve, Switzerland}
\author{Nicolas Brunner}
\affiliation{D\'epartement de Physique Th\'eorique, Universit\'e de Gen\`eve, 1211 Gen\`eve, Switzerland}

\begin{abstract}
Quantum steering can be detected via the violation of steering inequalities, which provide sufficient conditions for the steerability of quantum states. Here we discuss the converse problem, namely ensuring that an entangled state is unsteerable, and hence Bell local. We present a simple criterion, applicable to any two-qubit state, which guarantees that the state admits a local hidden state model for arbitrary projective measurements. Specifically, we construct local hidden state models for a large class of entangled states, which can thus not violate any steering or Bell inequality. In turn, this leads to sufficient conditions for a state to be only one-way steerable, and provides the simplest possible example of one-way steering. Finally, by exploiting the connection between steering and measurement incompatibility, we give a sufficient criterion for a continuous set of qubit measurements to be jointly measurable.
\end{abstract}

\maketitle
%
%
\section{Introduction}

Entanglement lies at the heart of quantum physics. Notably, correlations arising from local measurements performed on separated entangled systems can exhibit nonlocal correlations \cite{bell64,brunner_review}. Specifically, the observed statistics cannot be reproduced using a local hidden variable model, as witnessed by violation of a Bell inequality. 

Recently, the effect of Einstein-Podolsky-Rosen steering has brought novel insight to quantum nonlocality. Originally discussed by Schr{\"o}dinger \cite{schrodinger}, and used in quantum optics \cite{reid08}, the effect was recently formalized in a quantum information-theroretic setting \cite{wiseman07}. Considering two distant observers sharing an entangled state, steering captures the fact that one observer, by performing a local measurement on his subsystem, can nonlocally `steer' the state of the other subsystem. Steering can be understood as a form of quantum nonlocality intermediate between entanglement and Bell nonlocality \cite{wiseman07,quintino15}, and is useful to explore the relation between these concepts. It was demonstrated experimentally, see e.g. \cite{wittman}, and finds application in quantum information processing \cite{cyril,piani,valerio}.

Steering can be detected via steering inequalities (analogous to Bell inequalities) \cite{Cavalcanti09}, violation of which provides a sufficient condition for a given quantum state to be steerable. Derived for both discrete and continuous variable quantum systems \cite{Cavalcanti09,walborn,Cavalcanti15}, such inequalities can be obtained using semidefinite programming \cite{pusey13,paulsteering,Kogias15,zhu15}.

Interestingly, whereas the effect of steering implies the presence of entanglement, the converse does not hold \cite{wiseman07}. Specifically, there exist entangled states which can provably not give rise to steering (hence refereed to as unsteerable) \cite{werner89,wiseman07}, even when general measurements are considered \cite{quintino15}. The correlations of such states can in fact be reproduced without entanglement, using a so-called local hidden state (LHS) model \cite{wiseman07}, and therefore can never violate any steering inequality. Since a LHS model is a particular case of a local hidden variable model, any unsteerable state is Bell local.

Determining which entangled states are steerable and which ones are not is a challenging problem in general.  This is mainly due to the fact that, when constructing a LHS model, one must ensure that the model reproduces the desired quantum correlations for any possible measurements. LHS models have been constructed for entangled states featuring a high level of symmetry \cite{werner89,barrett02,almeida07,bowles14,bowles15}; see \cite{augusiak_review} for a review. For more general states very little is known, even for the simplest case of two-qubit states. Based on the concept of the steering ellipsoid \cite{ellipsoid}, 
Ref. \cite{sania14} derived a condition guaranteeing unsteerability of Bell diagonal two-qubit state. This method is however not applicable to general two-qubit states, for which unsteerability conditions are still missing.

Here, via the construction of a class of LHS models, we derive a simple criterion sufficient for guaranteeing that a two-qubit state is unsteerable, considering arbitrary projective measurements.
In turn, this criterion can also be used to guarantee one-way steerability \cite{Midgley10,bowles14}, a weak form of steering where only one of the observers can steer the state of the other. We illustrate the relevance of the criterion with examples, providing in particular the simplest possible example of one-way steering. Finally, by exploiting the strong connection between steering and measurement incompatibility \cite{quintino14,Uola14}, we provide a sufficient condition for a continuous set of dichotomic qubit POVMs to be jointly measurable.

%
%
\section{Preliminaries} 
Consider two distant parties, Alice and Bob, sharing an entangled quantum state $\rho$. On her subsystem, Alice makes measurements, described by operators $\{M_{a|x}\}$, with $M_{a|x}\geq 0$ and $\sum_a M_{a|x}= \openone$, where $x$ denotes the measurement setting and $a$ its outcome. The possible states of Bob's subsystem, conditioned on Alice's measurement $x$ and her output $a$, are characterized by a collection of (subnormalized) density matrices $\{\sigma_{a|x}\}_{a,x}$, called an assemblage, with
\begin{align}\label{assdef}
\sigma_{a|x} = \Tr_{A}(M_{a|x}\otimes \openone\;\rho).
\end{align}
Note that it also includes Alice's marginal statistics $p(a|x)=\Tr \sigma_{a|x}$. The assemblage $\{ \sigma_{a|x}\} $ is called unsteerable if it can be reproduced by a LHS model, i.e. it admits a decomposition 
\begin{align}\label{lhsass}
\sigma_{a|x}=\sigma^{\text{\tiny{LHS}}}_{a|x}=\int \sigma_{\lambda} \, p(a|x\lambda)\text{d}\lambda  \quad \forall a,x ,
\end{align}
where $\{\sigma_{\lambda}\}$ is a set of positive matrices such that $\int\Tr\sigma_{\lambda}\text{d}\lambda=1$ and the $p(a|x,\lambda)$'s are probability distributions. The right hand side of the above can be understood as follows. Alice sends the quantum state $\sigma_\lambda/\Tr\sigma_\lambda$ to Bob with probability density $\Tr\sigma_\lambda$. Given her measurement input $x$, she then outputs $a$ with probability $p(a|x,\lambda)$. 
In this way, Alice can prepare the same assemblage for Bob as if the state $\rho$ had been used, without the need for entanglement. Bob will thus be unable to distinguish whether he and Alice share the entangled state $\rho$ or if the above LHS strategy were used. On the contrary, if a decomposition of the form \eqref{lhsass} does not exist, which can be certified e.g. via violation of a steering inequality \cite{Cavalcanti09}, the use of an entangled state is certified. In this case, $\rho$ is termed steerable from Alice to Bob.

Interestingly, not all entangled states are steerable. That is, there exists entanged states, called unsteerable, which admit a LHS model for all projective measurements \cite{werner89,wiseman07} and even considering general POVMs \cite{quintino15}. A natural question is thus to determine which entangled states are steerable, and which ones are not. This is a challenging problem, mainly due to the difficulty of constructing LHS models for a continuous set of measurements.

%
%

\section{Sufficient criterion for unsteerability} Our main result is a simple criterion, sufficient for a two-qubit state to admit a LHS model for arbitrary projective measurements. Consider a general two-qubit state, expressed in the local Pauli basis
\small\begin{align}\label{rho1}
\rho_0=\frac{1}{4}\left(\openone+\vec{a}_0\cdot\vec{\sigma}\otimes \openone+\openone\otimes\vec{b}_0\cdot\vec{\sigma}+\sum_{i,j=x,y,z}T_{ij}^0\sigma_{i}\otimes\sigma_{j}\right),
\end{align}\normalsize
where $\vec{a}_0$ and $\vec{b}_0$ are Alice and Bob's local Bloch vectors and $\vec{\sigma}=(\sigma_x,\sigma_y,\sigma_z)$ is the vector of Pauli matrices. Our criterion for unsteerability is simply based on the local Bloch vectors and the correlation matrix $T_{ij}^0$. 

The first step consists in converting the state $\rho_0$ into a canonical form. For this, based on previous work \cite{rodrigo14,quintino15}, we make the observation that

\begin{lemma}
Let $\Lambda$ be a positive linear map on the set of quantum states and
\begin{align}
\rho_{\Lambda}=\openone\otimes \Lambda (\rho)/\Tr(\openone\otimes \Lambda (\rho))
\end{align}
be a valid bipartite quantum state. If $\rho$ is unsteerable from Alice to Bob, then $\rho_{\Lambda}$ is also unsteerable from Alice to Bob. Furthermore, if $\Lambda$ is invertible and its inverse map positive, then $\rho$ is unsteerable from Alice to Bob if and only if $\rho_{\Lambda}$ is unsteerable from Alice to Bob.
\end{lemma}

\noindent Note that $\Lambda$ does not have to be completely positive and may therefore correspond to a non-quantum operation. For a proof see e.g. Lemma 2 of \cite{quintino15}, where the condition of complete positivity can simply be relaxed.

Let us now consider the positive linear map 
\begin{align}\label{cmap}
\openone\otimes\Lambda(\rho_0)=\openone\otimes \rho_{B}^{-1/2}\; \rho_0 \; \openone\otimes \rho_{B}^{-1/2}
\end{align}
where $\rho_{B}=\Tr_{A}[\rho_0]$. This map is invertible as long as $\rho_B$ is mixed with the inverse (positive) map given by
\begin{align}\label{cmap}
\openone\otimes\Lambda^{-1}(\rho_0)=\openone\otimes \rho_{B}^{1/2}\; \rho_0 \; \openone\otimes \rho_{B}^{1/2}.
\end{align}
The interesting property of this map is that when applied to an arbitrary state $\rho_0$, the resulting state has $\vec{b}=0$, i.e. Bob's reduced state is maximally mixed \cite{milne14}. Given the above Lemma, the application of the map preserves the steerability (or unsteerability) of $\rho_0$.

Finally, we apply local unitaries (which can also not change the steerability of the state) so that our state has a diagonal $T$ matrix, giving us the canonical form
\begin{align}\label{canonical}
\rho=\frac{1}{4}\left(\openone+\vec{a}\cdot\vec{\sigma}\otimes \openone+\sum_{i=x,y,z}T_{i} \, \sigma_{i}\otimes\sigma_{i}\right),
\end{align}
where $\vec{a}$ and $T_{i}$ are in general different from the original $\vec{a}_0$ and $T_{ii}^0$. Below we give a sufficient criterion for the unsteerability of any state $\rho$ expressed in the canonical form. In turn this provides a sufficient criterion for unsteerability of any two-qubit state.

\begin{theorem}
Let $\rho_0$ be a two-qubit state with corresponding canonical form $\rho$ as given in \eqref{canonical}. If
\begin{align} \label{condition}
\max_{\hat{x}} \left[\;(\vec{a}\cdot\hat{x})^2+2\;||T\hat{x}|| \; \right] \leq 1,
\end{align}
where $\hat{x}$ is a normalized vector and $||\cdot||$ the euclidean vector norm, then $\rho$ is unsteerable from Alice to Bob, considering arbitrary projective measurements.
\end{theorem}

\begin{proof}
We first characterize the assemblage resulting from projective measurements on a state in the canonical form $\rho$. Alice's measurement is given by a Bloch vector $\hat{x}$ and output $a=\pm1$, corresponding to operators $ M_{\pm|\hat{x}} = (\openone + \hat{x} \cdot \vec{\sigma})/2$. For $a=+1$, the steered state is (see for example \cite{sania14})
\begin{align}\label{assparam}
\sigma_{+ |\hat{x}}=\Tr_A( M_{+ | \hat{x}}  \otimes \openone \rho) = \frac{1}{4}[(1+\vec{a}\cdot\hat{x})\openone+T\hat{x}\cdot\vec{\sigma}].
\end{align} 
Notice that the above state is diagonal in the basis $\{ \ket{\hat{s} }, \ket{-\hat{s}} \}$, with Bloch vector $\hat{s} = \frac{T\hat{x}}{|| T\hat{x} ||}$; we omit the $\hat{x}$ dependence to ease notation. The corresponding eigenvalues are
\small\begin{align} \label{digs}
\alpha(\hat{x})=\frac{1}{4}(1+\vec{a}\cdot\hat{x}+||T\hat{x}||)  , \,\,
\beta(\hat{x})=\frac{1}{4}(1+\vec{a}\cdot\hat{x}-||T\hat{x}||).
\end{align}\normalsize
Note that by construction $\alpha(\hat{x}) \geq \beta(\hat{x})$.

Our goal is now to construct a LHS model for this assemblage. First, the local hidden states $\sigma_\lambda$ are taken to be pure qubit states, hence represented by unit Bloch vectors $\hat{\lambda}$, and uniformly distributed over the sphere
\begin{align}
\sigma_{\hat{\lambda}}=\frac{\ket{\hat{\lambda}}\bra{\hat{\lambda}}}{4\pi}.
\end{align}
Normalization is ensured as $\int \Tr[\sigma_{\hat{\lambda}}]\text{d}\hat{\lambda}=1$. 
This ensures that we obtain the correct reduced state for Bob: 
\begin{align}\label{rhob}
\frac{1}{4\pi}\int \ket{\hat{\lambda}}\bra{\hat{\lambda}} \text{d}\hat{\lambda}=\frac{\openone}{2} = \rho_B.
\end{align}

Next, we define Alice's response function to be given by the distribution 
\begin{align}\label{fun}
p(\pm |\hat{x},\hat{\lambda})=\frac{1\pm\text{sgn}(\hat{s}\cdot \hat{\lambda}-c(\hat{x}))}{2},
\end{align}
parameterized by the function $-1\leq c(\hat{x})\leq1$, and $\hat{s}$ the Bloch vector of the eigenvector of $\sigma_{+ |\hat{x}}$ with the largest eigenvalue. The function \eqref{fun} can be understood as follows (see Fig.\ref{pfig}). If $\hat{\lambda}$ is in the spherical cap centred on $\hat{s}$ such that $\hat{\lambda}\cdot\hat{s} \geq c(\hat{x})$, the output is $a=+1$, otherwise $a=-1$. Note that we need only concentrate on the case $a=+1$;  the case $a=-1$ is automatically satisfied from $\sigma_{+1|\hat{x}}+\sigma_{-1|\hat{x}}=\rho_B$ and \eqref{rhob}. We now calculate the assemblage predicted by this model, given by
\begin{figure}
\includegraphics[height=145pt]{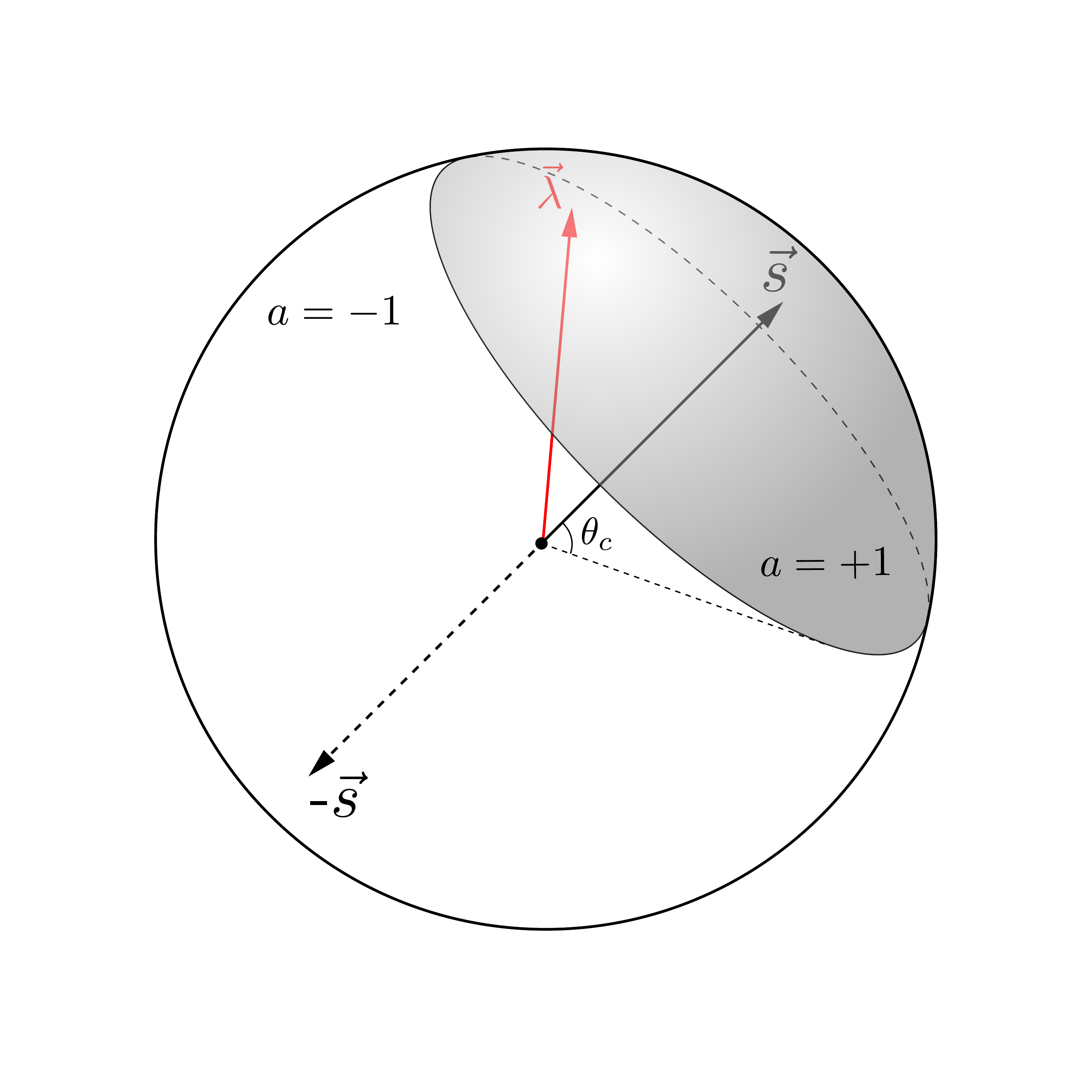}
\caption{\label{pfig} Illustration of Alice's response function \eqref{fun} in our LHS model. If $\text{sgn}(\hat{s}\cdot \hat{\lambda}-c(\hat{x}))\geq0$ then $a=+1$ (shaded spherical cap, with angle $\theta_c=\arccos[c]$), otherwise $a=-1$. The assemblage \eqref{asslhs} then corresponds to the average (sub-normalized) density matrix obtained by integrating pure qubit states $\ket{\hat{\lambda}}$ over the shaded region.}
\end{figure}

\begin{figure}
\includegraphics[scale=0.35]{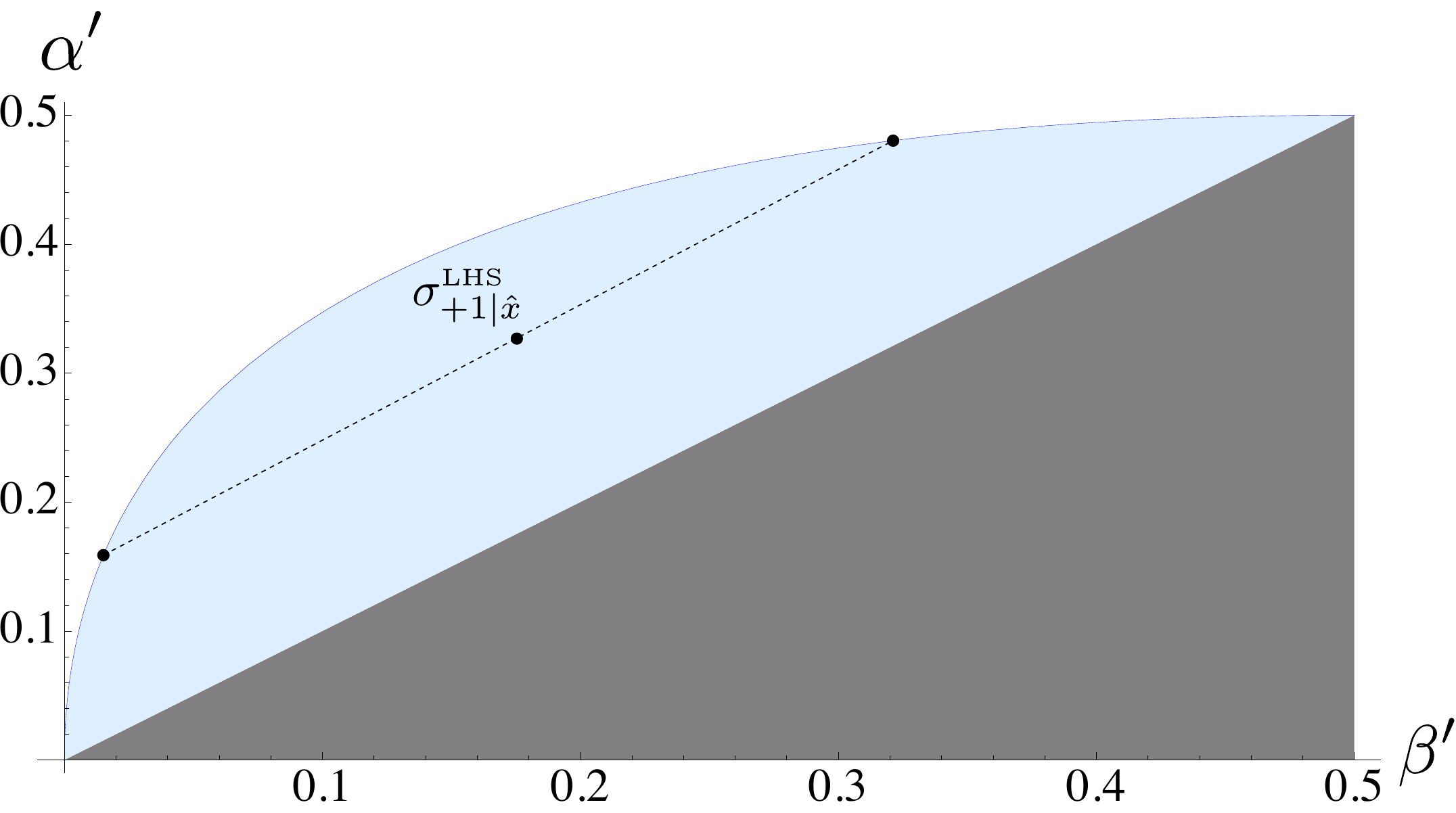}
\caption{\label{abfig}Plot of the achievable range of eigenvalues $(\alpha',\beta')$ in our LHS model (for a fixed direction $\hat{s}$). The upper blue curve corresponds to the condition $\alpha'=\sqrt{2\beta'}-\beta'$ and is achieved by the response functions \eqref{fun}; any point in the light blue area below may be achieved by taking a suitable convex combination of these functions (e.g. dashed line). Since we have $\alpha'\geq\beta'$, the grey area is not of interest.}
\end{figure}
\begin{align}\label{asslhs}
\sigma^{\text{LHS}}_{+1|\hat{x}}  &=  \int  \sigma_{\hat{\lambda}}\;p(+|\hat{x},\hat{\lambda}) d\hat{\lambda} = \frac{1}{4\pi} \int  \ket{\hat{\lambda}}\bra{\hat{\lambda}}p(+|\hat{x},\hat{\lambda}) d\hat{\lambda}. \nonumber
\end{align}
We parameterise the state $\ket{\hat{\lambda}}$ using the Bloch decomposition in the basis $\{\ket{\hat{s}},\ket{-\hat{s}}\}$:
\begin{align}
\ket{\hat{\lambda}}=\ket{\hat{\lambda}(\theta,\phi)}=\cos\frac{\theta}{2}\ket{\hat{s}}+\sin\frac{\theta}{2}e^{i\phi}\ket{-\hat{s}}.
\end{align}
Working in this basis and integrating over the spherical cap for which $a=+1$ (see Fig. \ref{pfig}), \eqref{asslhs} becomes
\begin{align}
\int_{0}^{2\pi}\int_{0}^{\theta_c}\begin{pmatrix}\cos^2\frac{\theta}{2} & \cos\frac{\theta}{2}\sin\frac{\theta}{2}e^{-i\phi} \\[3pt] \cos\frac{\theta}{2}\sin\frac{\theta}{2}e^{i\phi} &\sin^2\frac{\theta}{2} \end{pmatrix}\frac{\sin\theta\text{d}\phi\text{d}\theta}{4\pi}, \nonumber
\end{align}
where $\theta_c=\arccos[c(\hat{x})]$ is the angle of the spherical cap. Since $\int_0^{2\pi}e^{i\phi}\text{d}\phi=0$, the off-diagonal components will be zero, and $\sigma^{\text{\tiny{LHS}}}_{+1|\hat{x}}$ is therefore diagonal in the $\{\ket{\hat{s}},\ket{-\hat{s}}\}$ basis, as desired. From this, the eigenvalues of $\sigma^{\text{\tiny{LHS}}}_{+1|\hat{x}}$, $\alpha'(\hat{x})$ and $\beta'(\hat{x})$, are given by 
\begin{align}
\alpha'(\hat{x})+\beta'(\hat{x})&=\frac{1}{2}\int_0^{\theta_c}\sin\theta\text{d}\theta=\frac{1-\cos\theta_c}{2} ;\\
\alpha'(\hat{x})-\beta'(\hat{x})&=\frac{1}{2}\int_0^{\theta_c}\cos\theta\sin\theta\text{d}\theta=\frac{1-\cos^2\theta_c}{4}.
\end{align}
Upon using $\theta_c=\arccos[c(\hat{x})]$ one then finds 
\begin{align}
\alpha'(\hat{x})+\beta'(\hat{x})&=\frac{1}{2}(1-c(\hat{x})) ;  \\
\alpha'(\hat{x})-\beta'(\hat{x})&=\frac{1}{4}(1-c^2(\hat{x})),
\end{align}
from which we get the eigenvalues as a function of $c(\hat{x})$ as
\begin{align}\label{evalues}
\alpha'(\hat{x})=\sqrt{2\beta'(\hat{x})}-\beta'(\hat{x})  \, ; \,\, \beta'(\hat{x})=\frac{1}{8}(1-c(\hat{x}))^2,
\end{align}
corresponding to the curve of Fig. \ref{abfig}. Since this curve is concave, by fixing $\hat{s}$ and taking convex combinations of the response functions \eqref{fun} with different $c(\hat{x})$, we may prepare any steered states corresponding to $(\alpha',\beta')$ below this curve, leading finally to \begin{align}
\beta'(\hat{x}) \leq \alpha'(\hat{x}) \leq \sqrt{2\beta'(\hat{x})}-\beta'(\hat{x}).
\end{align} 
This corresponds to the blue area in Fig. \ref{abfig}. We thus conclude that the model reproduces the assemblage of any canonical state $\rho$, as long as its eigenvalues satisfy the above relation, i.e. $\alpha(\hat{x}) \leq \sqrt{2\beta(\hat{x})}-\beta(\hat{x})$, for any measurement vector $\hat{x}$, or equivalently  
\begin{align}
\max_{\vec{x}} \left[(\alpha(\vec{x})+\beta(\vec{x}))^2-2\beta(\vec{x})\right]\leq 0.
\end{align}
Using $\eqref{digs}$ to convert this maximisation into Bloch vector notation we arrive at \eqref{condition}.

\end{proof}

A natural question is whether condition \eqref{condition} is also necessary for unsteerability. Unfortunately, this is not the case. Consider the state $\rho_c = \frac{1}{2}(\ket{00}\bra{00}+\ket{11}\bra{11})$, which does not satisfy \eqref{condition} (choose e.g. $\hat{x}= (0,0,1)$), but is separable hence clearly unsteerable. Note however that condition \eqref{condition} can in fact be strengthened by considering convex combination with separable states (see Appendix A). An interesting open question is then whether there exist unsteerable states, which cannot be written as convex combinations of unsteerable states satisfying condition \eqref{condition} and separable states. Nevertheless, condition \eqref{condition} turns out to be useful for proving the unsteerability of interesting classes of states, as we illustrate below. 

\section{Applications}
We now illustrate the relevance of the above result with some applications. We will consider the class of states 
\begin{align}\label{mafalda}
\rho(p,\chi)=p\ket{\psi_\chi}\bra{\psi_\chi}+(1-p)\rho^{A}_{\chi}\otimes\openone/2
\end{align}
where $\ket{\psi_\chi}=\cos\chi\ket{00}+\sin\chi\ket{11}$ is a partially entangled two-qubit state, $\rho^{A}_{\chi}=\Tr_{B}\ket{\psi_\chi}\bra{\psi_\chi}$ and have $p\in[0,1]$, $\chi\in]0,\pi/4]$. The state is entangled for $p>1/3$. From Theorem 1, it follows that $\rho(p,\chi)$ is unsteerable from Alice to Bob if 
\begin{align}\label{mafcon}
\cos^2 (2\chi) \geq \frac{2p-1}{(2-p)p^3}
\end{align}
as we show in Appendix B. This result is illustrated in Fig.3 (black solid line). Note that our result recovers the case of a two-qubit Werner state, $\rho(1/2,\pi/4)$, which admits a LHS model \cite{werner89} (in both directions).

\subsection{One-way steering} Alice and Bob play different roles in the steering scenario. Hence steerability in one direction (say from Alice to Bob) does not necessarily imply steerability in the other direction (from Bob to Alice). This effect of one-way steering was first observed in the context of continuous variable systems and gaussian measurements \cite{Midgley10}. More recently, an example of a two-qubit one-way steerable state was presented considering arbitrary projective measurements \cite{bowles14}. That is, while Alice can steer Bob using a finite number of measurements, it would be impossible for Bob to steer Alice as the state admits a LHS model (Bob to Alice). Moreover, a qutrit-qubit state was shown to be one-way steerable considering POVMs \cite{quintino15}.

Clearly, our results are also useful for capturing one-way steering. Consider a given state $\rho$, the canonical form of which is found to satisfy condition \eqref{condition}. From Theorem 1, it follows that $\rho$ is unsteerable from Alice to Bob. Moreover, if steerability from Bob to Alice can be verified using standard methods, e.g. via violation of a steering inequality or using SDP methods \cite{pusey13,paulsteering}, one-way steerability of $\rho$ is proven. 

We present novel examples of one-way steering. Our states of interest will be the state $\rho(p,\chi)$ defined above. This state is unsteerable from Alice to Bob for projective measurements when \eqref{mafcon} is satisfied, corresponding to the area below the thick black line of Fig. \ref{mafcurve}. The steerability from Bob to Alice of the above state was discussed in previous works. In particular it was shown that $\rho(p,\chi)$ is unsteerable if $p\leq 1/2$ for all $\chi$ \cite{almeida07}. However, for $p>1/2$, the state becomes steerable from Bob to Alice for all $\chi$. This can be seen as follows. By applying on Alice's side the filter $F_\chi=\text{diag}(1/\cos\chi , 1/\sin\chi)$, we obtain the state 
\begin{align}\label{filterMaf}
\frac{1}{2}\;F_\chi\otimes\openone \, \rho(p,\chi) F_\chi\otimes\openone = \rho(p,\pi/4), 
\end{align}
which is simply a Werner state with visibility $p$. Since this state is steerable for $p>1/2$ \cite{wiseman07}, it follows from Lemma 1 that all states $\rho(p,\chi)$ with $p>1/2$ and satisfying \eqref{mafcon} are one-way steerable from Bob to Alice for projective measurements.

%

\begin{figure}
\includegraphics[height=155pt]{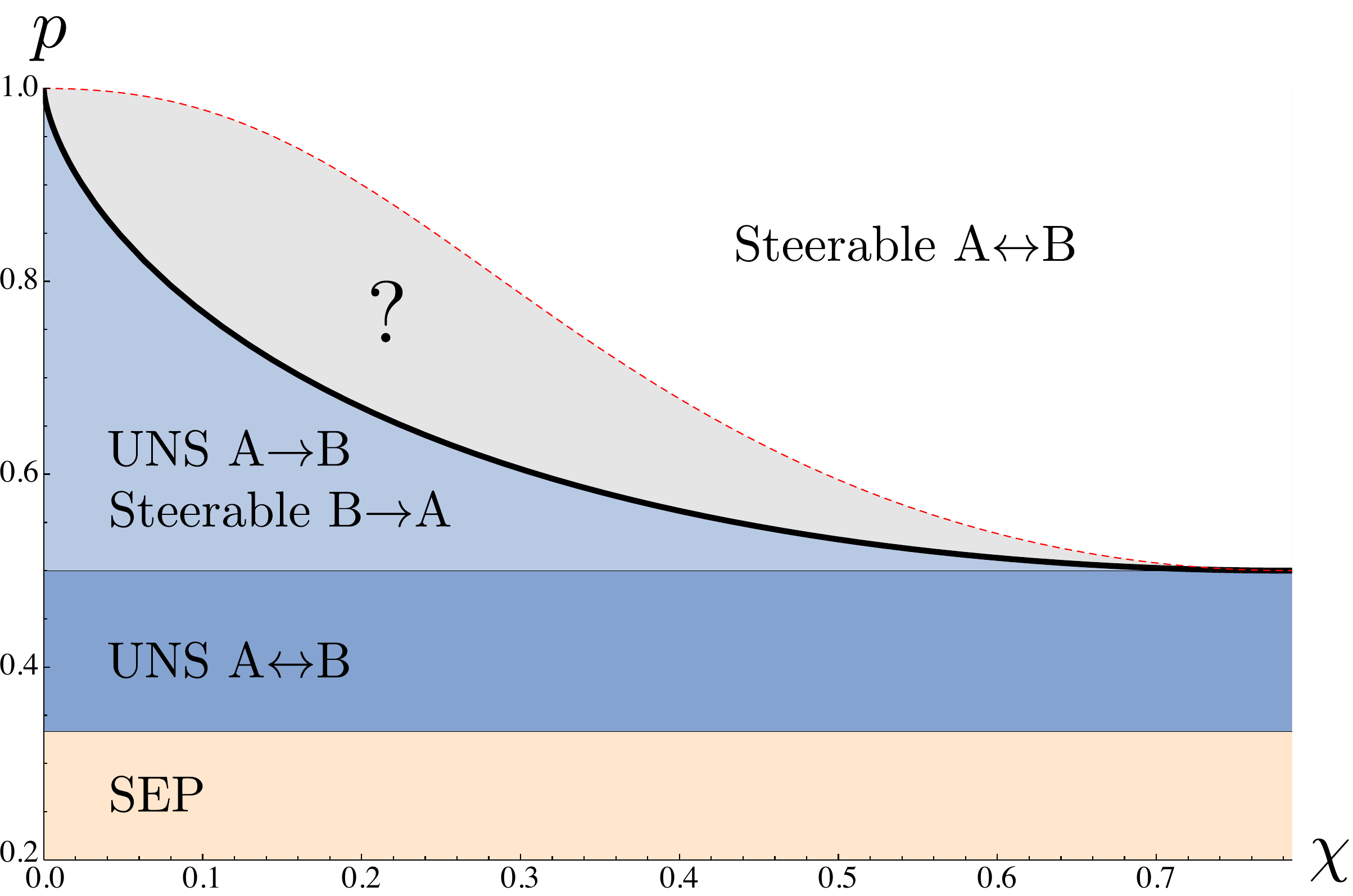}
\caption{\label{mafcurve} Characterization of entanglement and steering for states $\rho(p,\chi)$. The solid black curve corresponds to \eqref{mafcon}, obtained from our unsteerability criterion. The state is separable in the light orange region, unsteerable (in both directions) in the dark blue region, unsteerable only from Alice to Bob (hence one-way steerable) in the light blue region, and two-way steerable in the white region (obtained from equation 19 of \cite{sania14}). What happens in the grey region is an interesting open question. }
\end{figure}

\subsubsection{Simplest one-way steering}

A relevant question to ask is how many measurements are needed in order to demonstrate one-way steering. So far, the only known examples for a two-qubit state required as many as 13 measurements \cite{bowles14} and considered only projective measurements; similarly for the qutrit-qubit example of Ref. \cite{quintino15}.
Here we present the simplest possible example of one-way steering. That is, a two-qubit state such that Alice cannot steer Bob even with POVMs, although Bob can steer Alice using only two measurement settings.

We start with the case of projective measurements. We show that the states $\rho(p,\chi)$ with $p>1/\sqrt{2}$ and satisfying \eqref{mafcon} are one-way steerable, and only two measurements are required for demonstrating steering from Bob to Alice. To prove this we proceed as follows. First, from Lemma 1, it is sufficient to consider the state $\rho(p,\pi/4)$, i.e. a Werner state (see equation \eqref{filterMaf}). Since this state violates the CHSH Bell inequality for $p>1/\sqrt{2}$ \cite{brunner_review}, it is nonlocal, and thus steerable from Bob to Alice with two measurements.

Next, we move to the case of POVMs, building on the above example. Following protocol 2 of \cite{hirsch}, we construct the state
	\begin{equation}
	\rho_{\text{\tiny{POVM}}}(p,\chi)= \frac{1}{2} \rho(p,\chi) + \frac{1}{2} \ketbra{0}{0}\otimes \rho_B
	\end{equation} 
where $\rho_B=\Tr_A\rho(p,\chi)$, which is now unsteerable from Alice to Bob for POVMs, for $p$ and $\chi$ satisfying \eqref{mafcon}. We now show that steering from Bob to Alice is possible using only two measurements. From Lemma 1, we can focus our analysis on the state
	\begin{align}
	\rho_F&=\frac{F_\chi\otimes \openone\,\rho_{\text{\tiny{POVM}}} F_\chi\otimes \openone}{\Tr{(F_\chi\otimes \openone\,\rho_{\text{\tiny{POVM}}} F_\chi\otimes \openone)}}\\
	&= \frac{\cos^2\chi (p \ketbra{\phi^+}{\phi^+}+(1-p)\openone/4) +\frac{1}{2} \ketbra{0}{0}\otimes \rho_B}{\cos^2\chi+1/2}. \nonumber 
	\end{align}
Using the CHSH violation criterion \cite{chchcrit} one can find the range of parameters such that $\rho_{F}$ violates the CHSH inequality, and is thus steerable form Bob to Alice with two measurements. We find a parameter range $p>0.83353$ and corresponding $\chi$ given by condition \eqref{mafcon}.

%

\subsection{Sufficient condition for joint measurability} Theorem 1 also finds application in quantum measurement theory. This follows from the direct connection existing between steering and the notion of joint measurability of a set of quantum measurements \cite{quintino14,Uola14}, which has already found applications, see e.g. \cite{teiko}. This allows us to convert our sufficient condition for unsteerability into a sufficient condition for joint measurability of a set of qubit dichotomic POVMs. Notably, this condition is applicable to continuous sets of POVMs. 

A set of measurements $\{M_{a|x}\}$ is said to be jointly measurable \cite{buschbook} if there exists a joint POVM $\{G_{\lambda}\}$ with outcomes $\lambda$ and probability distributions $p(a|x,\lambda)$, from which the statistics of any of the measurements $\{M_{a|x}\}$ can be recovered by a suitable post processing, that is
\begin{align}
M_{a|x}=\int G_{\lambda}p(a|x,\lambda)\text{d}\lambda \quad\quad \forall\;a,x .
\end{align}
Let $\{M_{\pm|x}\}$ be a set of dichotomic qubit POVMs
\begin{align}
M_{+|x}=\frac{1}{2}\left(k_x\openone + \vec{m}_{x}\cdot\vec{\sigma}\right)
\end{align}
with $ ||\vec{m}_{x}|| \leq k_x \leq 2- ||\vec{m}_{x}||$, and $M_{-|x}= \openone - M_{+|x}$. Then the set $\{M_{\pm|x}\}$ is jointly measurable if 
\begin{align}\label{jmcon}
k_x(k_x-2)+2||\vec{m}_{x}||\leq0
\end{align}
for all $x$. This can be seen as follows. A set of measurements $\{M_{\pm|x}\}$ is jointly measurable if and only the assemblage given by $\sigma_{\pm|x}=\rho^{\frac{1}{2}}M_{\pm|x}\rho^{\frac{1}{2}}$, where $\rho$ is a full-rank quantum state \cite{UolaJM}, is unsteerable.
Choosing $\rho=\openone/2$ we get the corresponding assemblage $\sigma_{\pm|x}=\rho^{1/2}\;M_{\pm|x}\;\rho^{1/2}  = \frac{1}{2}M_{\pm|x}$. Following Theorem 1, condition \eqref{jmcon} ensures the unsteerability of $\sigma_{\pm|x}$, and consequently the joint measurability of $\{M_{\pm|x}\}$.

\section{conclusion} We have presented a simple criterion sufficient for a qubit assemblage to admit a LHS model. Notably, our method can guarantee the unsteerability of a general two-qubit state and should thus find applications. We have shown that the criterion allows one to detect entangled states which are only one-way steerable and provides the simplest such examples. Moreover, the criterion is relevant to quantum measurement theory, as it provides a sufficient condition for a continuous set of dichotomic qubit POVMs to be jointly measurable. Further to this, our criterion has also found applications in the connection between measurement incompatibility and nonlocality \cite{incompat} and multipartite nonlocality \cite{gme}. 

It would be interesting to extend this criterion in several directions. First, can the criterion be strengthened, e.g. by considering convex combinations, in order to become necessary and sufficient? Also, while we focused here on projective measurements, generalizing the method to POVMs would be useful \cite{footnote}. Whether the present ideas can be adapted to the case of higher dimensional systems (beyond qubits) is also a natural question. In particular, a natural case to consider is that of entangled states of dimension $d \times 2$, where our method should be directly applicable. Applications to multipartite steering \cite{he,cavalcanti} would also be interesting. 

\emph{Acknowledgements.}---We thank Charles Xu for discussions, and acknowledge financial support from the Swiss National Science Foundation (grant PP00P2\_138917 and Starting Grant DIAQ).

\begin{appendix}

\section{Convex combinations of unsteerable states}\label{sdp}

Since our criterion \eqref{condition} is not linear, and since it does not detect all separable states, it can be useful to consider convex combination of states. Specifically, consider an entangled unsteerable state of the form
\begin{align}
\rho=p\sigma+(1-p)\rho_{\text{\tiny{SEP}}},
\end{align} 
where $\rho_{\text{\tiny{SEP}}}$ is a separable (hence unsteerable) state, and $\sigma$ is an unspecified state. If $\sigma $ is unsteerable, then it follows that $\rho$ is unsteerable. However, it could be that, while $\rho$ violates condition \eqref{condition}, $\sigma$ does not. In this case, the unsteerability of $\rho$ can be shown by finding suitable $p$ and $\rho_{\text{\tiny{SEP}}}$ such that
\begin{align}\label{sepdecomp}
\sigma = \frac{\rho-(1-p)\rho_{\text{\tiny{SEP}}}}{p}
\end{align} 
satisfies condition \eqref{condition}. 

As a simple example, consider the state 
\begin{align}
\rho=\frac{1}{2}\sigma+\frac{1}{2}\left(\frac{1}{2}\ket{00}\bra{00}+\frac{1}{2}\ket{11}\bra{11}\right),
\end{align}
where $\sigma$ is the two-qubit isotropic state $\sigma=(\ket{\phi^+}\bra{\phi^+}+\openone/4)/2$. Hence $\rho$ is an equal mixture of $\sigma$ and the separable classically correlated state. One finds that for $\rho$, $T_{z}=3/4$ and so $\rho$ violates \eqref{condition} for $\hat{x}=(0,0,1)$. However, the state $\sigma$ has $T=\openone/2$ and $\vec{a}=\vec{0}$ and therefore satisfies \eqref{condition}, hence proving the unsteerability of $\rho$.

\section{Proof of unsteerability of $\rho(p,\chi)$}\label{e1proof}
Here we show that for the class of states \eqref{mafalda}, Theorem 1 implies that the $\rho(p,\chi)$ is unsteerable if 
\begin{align}\label{god}
\cos^2 2\chi \geq \frac{2p-1}{(2-p)p^3}.
\end{align}
To do this, we first consider states in canonical form \eqref{canonical}, which satisfy $\vec{a}=(0,0,a_z)$ and $|T_x|=|T_y|$. In order to perform the maximisation of Theorem 1, we parameterize $\hat{x}$ using spherical co-ordinates $\hat{x}= (\sin\theta\cos\phi,\sin\theta\sin\phi,\cos\theta)$. Our criterion \eqref{condition} may now be written as
\begin{align}
&\max_{\theta,\phi} F (\theta,\phi) \leq 1\,,\\ 
&F(\theta,\phi)=(\vec{a}\cdot\hat{x})^2+2||T\hat{x}|| \nonumber\\
&\quad\quad\quad=\cos^2\theta\;a_z^2+2\sqrt{T_x^2+\cos^2\theta\;(T_z^2-T_x^2)}. \nonumber
\end{align}
Unsurprisingly, $F$ depends only on $\theta$ since the problem is symmetric with respect to the $x$ and $y$ directions and we may ignore the maximisation over $\phi$. Note that if $|T_{z}|=|T_{x}|$ then the maximisation occurs at $\theta=0$ and our condition for unsteerability becomes 
\begin{align}\label{trivial}
a_z^2+2|T_z| \leq 1.
\end{align}
In the case $|T_{z}|\neq|T_{x}|$, one should find the extremal points of $F(\theta)$ and prove that they do not exceed 1. To find these extrema we solve
\begin{align}
\frac{\text{d}F}{\text{d}\theta}=-\sin2\theta\left( a_z^2+\frac{T_z^2-T_x^2}{\sqrt{T_x^2+\cos^2\theta(T_z^2-T_x^2)}}\right)=0.
\end{align}
From $\sin2\theta=0$ we have solutions $\theta=0,\pi/2,\pi$, and possibly other solutions given by
\begin{align}\label{nosol}
 a_z^2+\frac{T_z^2-T_x^2}{\sqrt{T_x^2+\cos^2\theta(T_z^2-T_x^2)}}=0.
\end{align}
We now derive conditions such that \eqref{nosol} has no solution. After rearranging \eqref{nosol} we have
\begin{align}
\cos^2\theta=\frac{T_x^2}{T_x^2-T_z^2}-\frac{T_x^2-T_z^2}{a^4_z}.
\end{align}
This has no solution if the RHS is greater than $1$ or less than $0$. Hence we have two conditions
\begin{align}\label{solcon}
\frac{T_x^2}{T_x^2-T_z^2}<\frac{T_x^2-T_z^2}{a^4_z}\quad\text{or}\quad
\frac{T_z^2}{T_x^2-T_z^2}>\frac{T_x^2-T_z^2}{a^4_z}.
\end{align}
If one of the above conditions is fulfilled we therefore have extrema for $\theta=0,\pi/2,\pi$ only. In this case, and since $F(0)=F(\pi)$, our condition for unsteerability becomes
\begin{align}\label{modcon}
\max_{\theta} F(\theta)=\max \{\;a_z^2+2|T_z| \;,\; 2|T_x|\;\}\leq 1.
\end{align}
We now move to the explicit case of $\rho(p,\chi)$. We find a canonical state with $|T_x|=|T_y|$, $\vec{a}=(0,0,a_z)$ and
\begin{align}
&a_z=\frac{(1-p^2)\cos2\chi}{1-p^2\cos^2 2\chi}\, ; \quad T_z=\frac{p(1-\cos^2 2\chi)}{1-p^2\cos^2 2\chi} \,; \nonumber \\[6pt]
&T_x=\sqrt{\frac{p^2(1-\cos^2 2\chi)}{1-p^2\cos^2 2\chi}}.
\end{align}
We now introduce the ansatz (for $p\geq\frac{1}{2}$)
\begin{align}\label{ansatz}
\cos^2 2\chi = \frac{2p-1}{(2-p)p^3}.
\end{align}
Eliminating the variable $\chi$ we find
\begin{align}
a^2_z&=\frac{(2-p)(2p-1)}{p}\,;\quad T_z=\frac{(1-p)^2}{p} \,;\nonumber\\
T_x&=1-p.
\end{align}
For the case $p=\frac{1}{2}$ we have $|T_{z}|=|T_{x}|$ and one finds that \eqref{trivial} is satisfied. For $p>\frac{1}{2}$ we show that the second condition of \eqref{solcon} holds. To this end, we calculate
\begin{align}
\frac{T_z^2}{T_x^2-T_z^2}-\frac{T_x^2-T_z^2}{a^4_z}=\frac{(3-p)(1-p)^3}{(p-2)^2(2p-1)}.
\end{align}
This is easily seen to be positive for $p\in]\frac{1}{2},1]$, and so $F(\theta)$ has extrema at $\theta=0,\pi,\pi/2$ only. It therefore remains to prove \eqref{modcon}. We find
\begin{align}
a_z^2+2|T_z|=1  \, ,\quad  2|T_x|=2(1-p).
\end{align}
and so \eqref{modcon} is satisfied for $p>\frac{1}{2}$. This proves that the state $\rho(p,\chi)$ is unsteerable if $p\geq\frac{1}{2}$ and $p$ and $\chi$ satisfy \eqref{ansatz}, which corresponds to the black curve of Fig. \ref{mafcurve} in the main text. Finally, we note that for a fixed $\chi$, lowering $p$ amounts to putting more weight on the separable part of the state. Since a convex combination of an unsteerable state with a separable state is also unsteerable, all points below the curve of  Fig. \ref{mafcurve} are also unsteerable. Hence, we arrive at \eqref{mafcon}. 

\end{appendix}


\begin{thebibliography}{10}
\bibitem{footnote} Note that, starting from our result, one can construct new entangled states admitting a LHS model for all POVMs, using Protocol 2 of Ref. \cite{hirsch}.
\end{thebibliography}
\end{document}